\newif\ifC
\newif\ifAnon
\Anontrue  

\ifC
  \ifAnon
    \documentclass[a4paper,UKenglish,cleveref, autoref, thm-restate, anonymous]{lipics-v2021}
  \else
    \documentclass[a4paper,UKenglish,cleveref, autoref, thm-restate]{lipics-v2021}
  \fi
\else
    \documentclass[11pt]{article}
\fi

\ifC
\else
\fi

\ifC
\author{Avi Kadria}
{Department of Computer Science, Bar-Ilan University, Ramat Gan 5290002, Israel}
{avi.kadria3@gmail.com}
{https://orcid.org/0000-0001-8449-3284}
{}

\authorrunning{A. Kadria}
\Copyright{Avi Kadria}



\EventEditors{}
\EventNoEds{1}
\EventLongTitle{Symposium on Theoretical Aspects of Computer Science (STACS 2027)}
\EventShortTitle{STACS 2027}
\EventAcronym{STACS}
\EventYear{2027}
\EventDate{March 8--12, 2027}
\EventLocation{G\"ottingen, Germany}
\EventLogo{}
\SeriesVolume{0}
\ArticleNo{0}
\ifAnon\hideLIPIcs\fi

\else
\date{}
\author{Avi Kadria\thanks{Department of Computer Science, Bar-Ilan University, Ramat Gan 5290002, Israel. E-mail: \texttt{avi.kadria3@gmail.com}.}}
\fi

\title{Deterministic Girth Approximation in Weighted Undirected Graphs}
\ifC\titlerunning{Deterministic Girth Approximation in Weighted Graphs}\fi
\usepackage{placeins}
\ifC
\else
\usepackage[margin=1in]{geometry}
\fi
\ifC\else
\usepackage{setspace}
\fi
\usepackage[bottom]{footmisc}

\usepackage{graphicx}
\usepackage{tikz}
\usetikzlibrary{arrows,automata,patterns}
\ifC\else\usepackage{enumitem}\fi
\usepackage{hyperref}
\usepackage{xspace}

\usepackage[linesnumbered,ruled,vlined,boxed,algo2e]{algorithm2e}
\usepackage{amsmath}
\usepackage{float}
\usepackage{mathtools}
\usepackage{amsthm}
\ifC\else
\usepackage{fullpage}
\fi
\usepackage{amsfonts}
\ifC
\else
\usepackage{thmtools}
\usepackage{thm-restate}
\usepackage[nameinlink]{cleveref}
\newtheorem{theorem}{Theorem}

\newtheorem{lemma}{Lemma}

\newtheorem{corollary}[theorem]{Corollary}

\fi

\usepackage{xcolor}
\usepackage{silence}
\begin{document}
\ActivateWarningFilters[pdftoc]

\newcommand{\liam}[1]{{\color{red} \textbf{Liam}: #1}}
\newcommand{\avi}[1]{{\color{purple} \textbf{Avi}: #1}}
\newcommand{\virgi}[1]{{\color{blue} \textbf{Virgi}: #1}}

\newcommand{\eps}{\varepsilon}
\newcommand{\Ot}{\tilde{O}}
\newcommand{\N}{\mathbb{N}}
\newcommand{\RR}{\mathbb{R}}

\newcommand{\codestyle}[1]{\texttt{#1}}
\newcommand{\Initialize}{\mbox{\codestyle{Initialize}}}
\newcommand{\Cycle}{\mbox{\codestyle{Cycle}}}
\newcommand{\ClusterOrCycle}{\codestyle{ClusterOrCycle}}
\newcommand{\ClusterOrCycleBounded}{\codestyle{ClusterOrCycleBounded}}
\newcommand{\BoundedCycle}{\mbox{\codestyle{BoundedCycle}}}
\newcommand{\Grow}{\mbox{\codestyle{Grow}}}

\newcommand{\Cl}{\mathit{CL}}


\DeclarePairedDelimiter{\ceil}{\lceil}{\rceil}
\DeclarePairedDelimiter{\floor}{\lfloor}{\rfloor}

\emergencystretch=3em 

\maketitle

\thispagestyle{empty}

\begin{abstract}
We give deterministic algorithms for approximating the girth~$g$ of an undirected graph with $n$ vertices, $m$ edges, and nonnegative real edge lengths. For every integer $k\geq2$, our algorithm returns a cycle of length at most $\frac{2k}{3}g$ in $O((m+n^{1+2/k})k\log n)$ worst-case time. For every constant~$k$, this matches the stretch and the expected running time of the randomized algorithms of Kadria, Roditty, Sidford, Vassilevska Williams, and Zwick [SODA'23, SICOMP'26] and Kadria, Roditty, and Vassilevska Williams [ESA'26]. In general, it loses only a factor of $O(\min\{k,\log n\})$ in the term that depends on~$m$. The best previous deterministic trade-off, of Kadria, Roditty, Sidford, Vassilevska Williams, and Zwick [SODA'22], has stretch~$k$ instead of $\frac{2k}{3}$ at the same exponent (for even~$k$) and requires integer edge lengths. In particular, for $k=3$ we obtain a deterministic $2$-approximation in $O((m+n^{5/3})\log n)$ time for real edge lengths, whereas the deterministic algorithm of Ducoffe [SIDMA'21] for real lengths has stretch $2+\varepsilon$. 

Our main technical idea is that a deterministic construction of the Thorup--Zwick hierarchy used by these algorithms may stop exploring as soon as it finds a cycle: the approximation analysis only needs clusters truncated at half the best cycle certificate found, so balls stopped in this way need not be hit. Together with a degree-weighted selection rule, this bounds the work of each level by its number of edges and stored ball entries. More generally, a ball need not be hit whenever the analysis only uses clusters up to a radius bounded by a certificate that the algorithm has already found. 
\end{abstract}
\ifC
\else
\clearpage
\pagenumbering{arabic}
\newpage
\fi

\section{Introduction}
Let $G=(V,E,\ell)$ be an undirected graph with $n$ vertices, $m$ edges, and nonnegative real edge lengths. The \emph{girth}~$g$ of~$G$ is the minimum length of a cycle in~$G$. Computing a shortest cycle is a fundamental problem that has been studied extensively, both in unweighted graphs (e.g.,~\cite{ItaiR78,AYZ97,YusterZ97,LingasL09,DBLP:conf/soda/RodittyW12,kadria2022algorithmic}) and in weighted graphs (e.g.,~\cite{LingasL09,DBLP:conf/focs/RodittyW11,RodittyT13,DBLP:journals/siamdm/Ducoffe21,kadria2023improved}). In weighted graphs, computing the girth exactly is equivalent to All-Pairs Shortest Paths under subcubic reductions~\cite{DBLP:conf/focs/RodittyW11,DBLP:journals/jacm/WilliamsW18}, so research has focused on approximation. 
We say an algorithm has \emph{stretch}~$c$ if it returns a cycle of length at most $c\cdot g$.

Lingas and Lundell~\cite{LingasL09} presented a $2$-approximation algorithm for integer lengths in $[1,M]$ that runs in $\Ot(n^2\log M)$ time.\footnote{$\Ot$ hides polylogarithmic factors, assuming $k=O(\log n)$; larger values of~$k$ only increase the stretch, since $n^{1+2/k}=O(n)$ already for $k\geq2\log n$.} Roditty and Tov~\cite{RodittyT13} improved the stretch to $4/3$ in the same $\Ot(n^2\log M)$ time.\footnote{They also obtained stretch $4/3+\eps$ in $\Ot(n^2/\eps)$ time for real lengths.} Ducoffe~\cite{DBLP:journals/siamdm/Ducoffe21} gave the first subquadratic approximation for weighted graphs: a $2$-approximation in $\Ot(n^{5/3}\operatorname{polylog}M+m)$ time for integer lengths in $[1,M]$.\footnote{In addition, he presented a $(2+\eps)$-approximation in $\Ot(n^{5/3}\operatorname{polylog}(1/\eps)+m)$ time for nonnegative real lengths.}

Kadria, Roditty, Sidford, Vassilevska Williams, and Zwick~\cite{kadria2022algorithmic} gave the first trade-off, which is also the best known deterministic one: for every integer $\kappa\geq1$, a cycle of length at most $2\kappa g$ can be found in $O((m+n^{1+1/\kappa}\log n)\log(nM))$ time for integer lengths in $[1,M]$.

Using randomization, this trade-off was improved in two steps. First, Kadria, Roditty, Sidford, Vassilevska Williams, and Zwick~\cite{kadria2023improved,kadria2026improvedsicomp} improved the stretch to $4\kappa/3$: for every integer $\kappa\geq1$, a cycle of length at most $4\kappa g/3$ can be found in $O((m+\kappa n^{1+1/\kappa})\log n+\kappa m)$ expected time. Then, Kadria, Roditty, and Vassilevska Williams~\cite{kadria2026tighter} extended this trade-off to stretch $2k/3$ in $\Ot(m+n^{1+2/k})$ expected time for every integer $k\geq2$.
The additive~$m$ term is necessary, since in weighted graphs one needs to read the input even for an approximation.\footnote{Take a complete graph whose edges have a huge length~$W$, and reduce the lengths of the edges of some triangle to~$1$. An algorithm with stretch~$c$ and $W>3c$ must find all three edges of length~$1$, since every other cycle has length more than~$3c$, and finding $3$ needles among $\Theta(n^2)=\Theta(m)$ edges requires $\Omega(m)$ time.}

Like the classic distance oracles of Thorup and Zwick~\cite{DBLP:journals/jacm/ThorupZ05}, these randomized algorithms sample a hierarchy of vertex sets $V=A_0\supseteq A_1\supseteq\cdots$, and random sampling guarantees that the clusters of this hierarchy, which the algorithms explore, have small expected total size. 
Their approximation analysis, however, applies to every nested hierarchy with the desired properties, so randomness is used only to bound the running time. 
The question is therefore whether the randomized running times can be matched deterministically, while keeping the same running time and stretch tradeoff.

\paragraph{Our results.} In this paper, we answer this question in the affirmative: our algorithms are deterministic, achieve the same stretch, and match the expected running times for every constant~$\kappa$. Formally, we prove the following two theorems.

\begin{theorem}\label{thm:deterministic-girth}
There is a deterministic algorithm that, given an integer $\kappa\geq1$ and a simple undirected graph $G=(V,E,\ell)$ with nonnegative real edge lengths that contains a cycle, returns a simple cycle of length at most $4\kappa g/3$, where $g$ is the girth of~$G$. Its running time is 
\[
 O\!\left(\bigl(m+n^{1+1/\kappa}\bigr)\kappa\log n\right).
\]
\end{theorem}

\begin{theorem}\label{thm:deterministic-2k3}
There is a deterministic algorithm that, given an integer $k\geq2$ and a simple undirected graph $G=(V,E,\ell)$ with nonnegative real edge lengths that contains a cycle, returns a simple cycle of length at most $2kg/3$, where $g$ is the girth of~$G$. Its running time is
\[
 O\!\left(\bigl(m+n^{1+2/k}\bigr)k\log n\right).
\]
\end{theorem}

\Cref{thm:deterministic-girth} is the case of even $k=2\kappa$ of \Cref{thm:deterministic-2k3}. We state it separately because its proof contains the core construction (\Cref{sec:truncated-clusters} to~\Cref{sec:deterministic-completion}); \Cref{sec:odd} then adapts the half-level construction of~\cite{kadria2026tighter} to obtain the odd values of~$k$. 

For every constant~$\kappa$, the bound of \Cref{thm:deterministic-girth} matches the expected running time $O((m+\kappa n^{1+1/\kappa})\log n+\kappa m)$ of~\cite{kadria2023improved}. Similarly, the bound of \Cref{thm:deterministic-2k3} matches the bound $O((m+kn^{1+2/k})\log n+km)$ of~\cite[Lemma~9]{kadria2026tighter}. 
For general~$\kappa$, the term that depends on~$n$ is the same, and the term that depends on~$m$ is larger by a factor of $O(\min\{\kappa,\log n\})$, since each of the $\kappa$ levels of our construction scans the edges with $O(\log n)$ heap operations; this loss disappears for constant~$\kappa$ and when $m=O(n^{1+1/\kappa})$. \Cref{tab:results} compares the bounds. Like the algorithms of~\cite{kadria2023improved,kadria2026tighter}, our algorithms only add and compare lengths, so they work for arbitrary nonnegative real lengths; unlike them, they use no hashing. 

\begin{table}[t]
\centering
\small
\renewcommand{\arraystretch}{1.15}
\begin{tabular}{@{}lllll@{}}
\hline
Algorithm & Stretch & Running time & Type & Lengths\\
\hline
\cite{kadria2022algorithmic} & $2\kappa$ & $O((m+n^{1+1/\kappa}\log n)\log(nM))$ & deterministic & integers in $[1,M]$\\
\cite{kadria2023improved,kadria2026improvedsicomp} & $4\kappa/3$ & $O((m+\kappa n^{1+1/\kappa})\log n+\kappa m)$ & randomized & nonnegative reals\\
\cite{yan2026lossless}+\cite{kadria2023improved} & $4\kappa/3$ & $O((m+n\log n)\kappa n^{1/\kappa})$ & deterministic & nonnegative reals\\
\Cref{thm:deterministic-girth} & $4\kappa/3$ & $O((m+n^{1+1/\kappa}) \kappa\log n)$ & deterministic & nonnegative reals\\
\hline
\cite{kadria2026tighter} (odd $k$) & $2k/3$ & $O((m+kn^{1+2/k})\log n+km)$ & randomized & nonnegative reals\\
\Cref{thm:deterministic-2k3} & $2k/3$ & $O((m+n^{1+2/k})k\log n)$ & deterministic & nonnegative reals\\
\hline
\cite{DBLP:journals/siamdm/Ducoffe21} ($k=3$) & $2+\eps$ & $\Ot(m+n^{5/3}\operatorname{polylog}(1/\eps))$ & deterministic & nonnegative reals\\
\hline
\end{tabular}
\caption{Weighted girth approximation; $\kappa\geq1$ and $k\geq2$ are integers. For even~$k$, the second block is covered by the first with $\kappa=k/2$. The bounds of~\cite{kadria2023improved,kadria2026tighter} hold in expectation, and those of the deterministic algorithms in the worst case.}

\label{tab:results}
\end{table}

\paragraph{Techniques.}
The Thorup--Zwick hierarchy can be derandomized by computing, deterministically, a small set that hits a ball of nearest centers around every vertex~\cite{DBLP:conf/icalp/RodittyTZ05,yan2026lossless}. However, computing these balls by shortest-path explorations scans the adjacency list of a vertex once for every center in its ball, which multiplies the edge term by the ball size, about $n^{1/\kappa}$. For example, combining the derandomization of Yan~\cite[Theorem~3]{yan2026lossless} with the algorithm of~\cite{kadria2023improved} gives $O((m+n\log n)\kappa n^{1/\kappa})$ time, which is $\Theta(n^2)$ instead of $\Ot(n^{3/2})$ for $m=n^{3/2}$ and $\kappa=2$. The greedy schedule of Yan~\cite{yan2026lossless}, which grows the balls in rounds with doubling quotas, controls the sizes of the balls, and the cluster searches of~\cite{kadria2023improved} explore clusters lazily. Neither of them, however, constructs the balls in a batch that can be interrupted and resumed. Our new ingredients are the following. 
\begin{enumerate}
 \item \emph{Truncation} (\Cref{lem:truncation-transfer}). Given a cycle certificate of value~$U$, the cluster searches may stop at radius $U/2$, so a ball need not be hit if a short cycle was already found.
 \item \emph{Safe termination} (\Cref{lem:batched-growth}). One batched exploration grows all balls. The first cycle it closes ends the current growth call and lowers the cutoff $U/2$ to at most the distance already explored; later rounds and levels continue with the smaller cutoff. Balls that stop early still contain everything the truncated analysis needs.
 \item \emph{Lazy propagation} (\Cref{sec:implementation}). Every propagation we process creates an entry, follows a shortest-path tree edge, or closes the cycle that ends the growth call. No propagation is processed twice, and a degree-weighted variant of Yan's selection rule makes the total degree of the growing balls drop geometrically.
\end{enumerate}

\paragraph{Determinism and related work.}
The randomized algorithms always return a correct cycle, but bound their running time only in expectation; our bounds hold in the worst case, which is also what is needed when a girth approximation is used as a subroutine of another deterministic algorithm. Compared with the previous deterministic trade-off~\cite{kadria2022algorithmic}, we improve the stretch from $2\kappa$ to $4\kappa/3$ at the same exponent, for nonnegative real lengths. Sampling-based shortest-path algorithms have been derandomized for distance oracles and spanners~\cite{DBLP:conf/icalp/RodittyTZ05,yan2026lossless}, single-source shortest paths~\cite{yan2026lossless}, and diameter approximation~\cite{kirkpatrick2026diameter}. Unlike these derandomizations, ours does not hit every ball, since a ball may stop at a radius set by a cycle certificate found during the construction.

\paragraph{Organization.}
\Cref{sec:preliminaries} introduces notation and the previous tools we use. \Cref{sec:overview} gives an overview of our techniques.  \Cref{sec:deterministic} proves \Cref{thm:deterministic-girth} and, in \Cref{sec:odd}, \Cref{thm:deterministic-2k3}.

\section{Preliminaries}
\label{sec:preliminaries}
Let $G=(V,E,\ell)$ be a simple undirected graph with $n$ vertices, $m$ edges, and nonnegative real edge lengths. As in~\cite{kadria2026improvedsicomp}, we assume that each addition or comparison of lengths takes $O(1)$ time. Zero lengths need no special care: we add an infinitesimal $\eps>0$ to every length, that is, we compare lengths as pairs $(\ell(e),1)$ lexicographically. Our algorithms only add and compare lengths, so they run unchanged, and all lengths become positive. A cycle of perturbed length at most $c$ times the perturbed girth has length at most $c\cdot g$, since the first coordinate of the perturbed girth is~$g$. Hence we may assume that all lengths are positive. The graph is given by adjacency lists, which can be sorted by length in $O(m\log n)$ time. 
We assume that $G$ is connected and contains a cycle; otherwise we find the connected components and discard the acyclic ones in $O(m+n)$ time, run the algorithm on each component, and return the shortest cycle found. All cycles we consider are simple. The girth~$g$ is the minimum cycle length, and $M(C)=\max_{e\in C}\ell(e)$ for a cycle~$C$. A \emph{cycle certificate} describes a cycle~$Z$ and a value $U\geq\ell(Z)$. Our algorithms maintain the smallest certificate value~$U$ found so far, initially $U=+\infty$ with no certificate.

Let $\delta(u,v)$ denote the shortest-path distance and $P(u,v)$ a shortest path from~$u$ to~$v$, and put $\delta(v,A)=\min_{a\in A}\delta(v,a)$, with $\delta(v,\emptyset)=+\infty$. We write $D[a,v]$ for a distance from~$a$ to~$v$ stored by an algorithm. Whenever vertices are compared by their distance from a vertex~$v$, ties are broken in favor of the smaller identifier.

We use the cluster graphs of~\cite{kadria2023improved,kadria2026improvedsicomp}. For an edge $e=\{v,w\}$, let $\delta(u,e)=\min\{\delta(u,v),\delta(u,w)\}+\ell(e)$. The \emph{radius-$r$ subgraph} $G_r(u)$ contains the vertices and edges at distance at most~$r$ from~$u$; $G_{<r}(u)$ uses strict inequalities. 

Given a hierarchy $V=A_0\supseteq\cdots\supseteq A_\kappa=\emptyset$, let $p_i(v)$ be the closest vertex of~$A_i$ to~$v$ (with the smallest identifier among the closest ones), whenever $A_i\neq\emptyset$. For $a\in A_i\setminus A_{i+1}$, its cluster graph $\Cl(a)=(\Cl_V(a),\Cl_E(a))$ is defined by
\begin{align*}
 \Cl_V(a)&=\{v:\delta(a,v)<\delta(v,A_{i+1})\},\\
 \Cl_E(a)&=\{\{v,w\}\in E:\delta(a,v)+\ell(v,w)<\delta(w,A_{i+1})
                 \text{ in at least one orientation}\}.
\end{align*}
We call an orientation $(v,w)$ satisfying this inequality a \emph{cluster orientation}. Every shortest path from~$a$ to a vertex of $\Cl_V(a)$ lies in $\Cl(a)$, and its edges, oriented away from~$a$, are cluster orientations~\cite[Lemma~2]{kadria2026improvedsicomp}.

\paragraph{Cluster searches.}
Our main subroutine is $\ClusterOrCycle$~\cite[Algorithm~2 and Lemma~5]{kadria2026improvedsicomp}, which runs Spira's lazy variant of Dijkstra's algorithm~\cite{Spira73} on $\Cl(a)$. 

\begin{lemma}[{\cite[Algorithm~2 and Lemma~5]{kadria2026improvedsicomp}}]\label{lem:cluster-or-cycle}
Let $a\in A_i\setminus A_{i+1}$. The search $\ClusterOrCycle(a)$ extracts cluster orientations $(v,w)$ of $\Cl(a)$ from a heap in nondecreasing order of their \emph{key} $D[a,v]+\ell(v,w)$. Each discovered vertex inserts its cluster orientations into the heap one at a time in order of length: the first one when the vertex is discovered, and the next one when the previous one has been extracted. When an orientation $(v,w)$ is extracted, one of the following holds.
\begin{enumerate}
 \item If $w$ is new, then $(v,w)$ becomes a tree edge and $D[a,w]=D[a,v]+\ell(v,w)=\delta(a,w)$.
 \item If $w$ is the predecessor of~$v$, then the orientation is the reverse of a tree edge and is skipped.
 \item Otherwise, the search stops and reports a certificate of value $D[a,v]+\ell(v,w)+D[a,w]$.
\end{enumerate}
If the heap becomes empty, the search returns the shortest-path tree of $\Cl(a)$. The search takes $O(|\Cl_V(a)|\log n)$ time.
\end{lemma}

An edge-access structure built during preprocessing produces the next cluster orientation in $O(\log n)$ time~\cite[Section~3.1]{kadria2026improvedsicomp}. Case~2 is not stated explicitly in~\cite[Algorithm~2]{kadria2026improvedsicomp}; we add it so that the reverse of a tree edge is not reported as a cycle.

\paragraph{Algorithm \texorpdfstring{$\Cycle$}{Cycle}.}
The algorithm $\Cycle$ of~\cite[Algorithm~4]{kadria2026improvedsicomp} samples the hierarchy, computes $p_i(v)$ and $\delta(v,A_i)$ by multi-source Dijkstra, runs $\ClusterOrCycle(a)$ for every vertex~$a$, and finally performs the \emph{edge-and-pivot checks}: for every edge $\{v,w\}$, in both orientations, and every nonempty~$A_i$, it considers the certificate of value $D[a,v]+\ell(v,w)+D[a,w]$ for $a=p_i(v)$, whenever $\{v,w\}$ is a non-tree edge whose endpoints both lie in the tree of~$a$.\footnote{In~\cite{kadria2026improvedsicomp}, the checks may also use the distances $\delta(p_i(v),v)$ computed during initialization; \Cref{lem:krsvz-analysis} shows that the checks above suffice.} It returns the best certificate found. Its only randomness is in the sampling of the hierarchy and in the hash tables that store distances.

We use the approximation analysis of~\cite{kadria2023improved} in the following form. For a shortest cycle~$C$ and $-1\leq i<\kappa$, define the radii
\begin{equation}\label{eq:analysis-radii}
 \Delta_i=\begin{cases}
  2(i+1)g/3 & \text{if } M(C)\leq g/3,\\
  (i+1)(g-M(C)) & \text{if } M(C)>g/3.
 \end{cases}
\end{equation}
In both cases $\Delta_{-1}=0$ and $\Delta_i\leq2(i+1)g/3\leq2\kappa g/3$ for $i<\kappa$, since $\min(2M(C), g-M(C))\le 2g/3$. The cycle~$C$, the girth~$g$, and the radii~$\Delta_i$ are used only in the analysis; the algorithms do not know~$g$. We call the cluster search for $a\in A_i\setminus A_{i+1}$ \emph{good} if
\begin{enumerate}
 \item[(C1)] it reports a certificate of value at most $2\Delta_i$, or
 \item[(C2)] it returns a predecessor tree~$T_a$ rooted at~$a$ that contains, as an edge, every cluster orientation $(v,w)$ of $\Cl(a)$ with $\delta(a,v)+\ell(v,w)\leq\Delta_i$.
\end{enumerate}
If (C2) holds and $x\in\Cl_V(a)$ satisfies $\delta(a,x)\leq\Delta_i$, then the edges of $P(a,x)$, oriented away from~$a$, are cluster orientations whose keys are the distances of their heads, at most~$\Delta_i$. Hence they all belong to~$T_a$, so $T_a$ contains~$x$ and its root-to-$x$ distance is $D[a,x]=\delta(a,x)$.

\begin{lemma}[Approximation analysis of~\cite{kadria2023improved}]\label{lem:krsvz-analysis}
Fix any hierarchy $V=A_0\supseteq\cdots\supseteq A_\kappa=\emptyset$ and a shortest cycle~$C$. If every cluster search is good, then the certificates reported by the cluster searches and by the edge-and-pivot checks on the returned trees include one of value at most $4\kappa g/3$.
\end{lemma}
\begin{proof}
The proof of~\cite[Lemmas~6--9 and Corollaries~7.1 and~9.1]{kadria2026improvedsicomp} applies, with every use of~\cite[Lemma~5]{kadria2026improvedsicomp} replaced by (C1) and~(C2). That proof shows, by induction on~$i$, that either some certificate has value at most $4(i+1)g/3$, or a vertex of~$C$ (if $M(C)\leq g/3$) or an endpoint of a longest edge $\{u,u'\}$ of~$C$ (if $M(C)>g/3$) is within distance~$\Delta_i$ of~$A_{i+1}$. Since $A_\kappa=\emptyset$, the first alternative must eventually hold.

\cite[Lemma~5]{kadria2026improvedsicomp} is used only for the search from a center $a\in A_i\setminus A_{i+1}$ in the inductive step, where the relevant edges of~$C$ have cluster orientations of key at most~$\Delta_i$ at~$a$~\cite[Lemmas~6 and~8]{kadria2026improvedsicomp}. If $M(C)\leq g/3$, these orientations form all of~$C$, so (C2) would put~$C$ into the tree~$T_a$; hence (C1) holds and gives a certificate of value at most~$2\Delta_i$. If $M(C)>g/3$, they form the path $C\setminus\{u,u'\}$. If (C1) fails, then (C2) puts this path into~$T_a$, so $\{u,u'\}$ is not a tree edge, $D[a,u]$ and $D[a,u']$ are exact by the remark above, and the edge-and-pivot check of $(u,u')$ at the pivot $a=p_i(u)$ yields a certificate of value at most $2\Delta_{i-1}+g\leq4(i+1)g/3$, as in~\cite[Lemma~8]{kadria2026improvedsicomp}.
\end{proof}

\paragraph{Cycle extraction.}
We use the standard cycle representation of~\cite[Section~2.5]{kadria2026improvedsicomp}. Let $T_a$ be a shortest-path predecessor tree rooted at~$a$, with root-to-$x$ distance $D[a,x]$, and let $P_{T_a}(x,y)$ denote its unique tree path. A non-tree edge $e=\{u,v\}$ with both endpoints in~$T_a$ defines, for $w=\operatorname{LCA}_{T_a}(u,v)$, the simple cycle
\begin{equation}\label{eq:fundamental-cycle}
 Z=P_{T_a}(w,u)\cup\{e\}\cup P_{T_a}(w,v),\qquad
 \ell(Z)=D[a,u]+\ell(e)+D[a,v]-2D[a,w].
\end{equation}
The cycle~$Z$ can be extracted in $O(n)$ time from the predecessor chains of~$u$ and~$v$~\cite[Section~2.5]{kadria2026improvedsicomp}.

\section{Technical overview}
\label{sec:overview}

We describe the construction for \Cref{thm:deterministic-girth}.

\paragraph{The randomized algorithm.}
For a hierarchy $V=A_0\supseteq\cdots\supseteq A_\kappa=\emptyset$, the \emph{bunch} $B_i(v)$ of a vertex~$v$ consists of the centers of~$A_i$ that are closer to~$v$ than every center of~$A_{i+1}$, that is, the centers whose level-$i$ cluster contains~$v$. By double counting, the total size of the level-$i$ clusters is $\sum_v|B_i(v)|$. Besides $O((m+\kappa n)\log n)$ time for preprocessing and $O(\kappa m)$ time for the edge-and-pivot checks, the running time of~\cite{kadria2023improved} is $O(\log n)$ times the total cluster size, which is $O(\kappa n^{1+1/\kappa})$ in expectation if each center of~$A_i$ is sampled into~$A_{i+1}$ with probability $n^{-1/\kappa}$.

\paragraph{Choosing a level deterministically.}
Fix $A=A_i$. For every vertex~$v$, its \emph{ball} $L(v)$ is a list of the centers of~$A$ nearest to~$v$, in order of distance. If $L(v)$ contains a center of the next level $R=A_{i+1}$, then $B_i(v)\subseteq L(v)$. The greedy schedule of Yan~\cite{yan2026lossless} chooses~$R$ in rounds. In each round, it grows every unhit ball to a quota of~$b$ centers, and greedily selects centers until at most a quarter of the full balls remain unhit. It then stops growing the hit balls and doubles~$b$. Starting with a quota of $\Theta(n^{1/\kappa})$, it selects at most $|A|/n^{1/\kappa}$ centers. Since the number of unhit balls drops by a factor of four whenever the quota doubles, the balls store $O(n^{1+1/\kappa})$ entries in total. We use a weighted variant of this rule that also makes the total degree of the unhit balls drop by a factor of four.

\paragraph{The obstacle: growing the balls.}
Balls are grown by a simultaneous shortest-path exploration from all centers: an entry for center~$a$ stored at vertex~$u$ is propagated along the edges of~$u$. Propagating every entry along every edge costs $\sum_u\deg(u)|L(u)|$, which can be $\Theta(mn^{1/\kappa})$.

We use cycles to avoid this. Suppose the entry of~$a$ at~$u$ is propagated along $(u,v)$ at distance $s=D[a,u]+\ell(u,v)$, and $v$ already stores~$a$ at distance $D[a,v]\leq s$. Then either $\{u,v\}$ is an edge of the shortest-path tree of~$a$, or this tree together with $\{u,v\}$ contains a cycle of length at most $D[a,u]+\ell(u,v)+D[a,v]\leq2s$. Hence every propagation creates an entry, is one of at most two propagations along a tree edge per entry, or finds a cycle, and we stop the exploration at the first cycle. We use a lazy event queue that never touches an edge before it is needed and never processes an event twice. Together with the weighted selection rule, this bounds the cost of all rounds of a level by $O((m+n^{1+1/\kappa})\log n)$.

\paragraph{Truncation.}
Stopping at a cycle leaves balls that are neither full nor hit. Let $U$ be the smallest certificate value found so far. Our key observation (\Cref{lem:truncation-transfer}) is that if $U>4\kappa g/3$, the analysis of~\cite{kadria2023improved} uses cluster information only up to distance $\tfrac23\kappa g<U/2$ from each center; if $U\leq4\kappa g/3$, the cycle found already suffices. It is therefore enough to bound the \emph{truncated} bunches, consisting of the centers $a\in B_i(v)$ with $\delta(a,v)<U/2$; dually, each cluster is explored only up to distance $U/2$ from its center (\Cref{fig:cluster-truncation}). 

A cycle found at distance~$s$ gives $U\leq2s$, so the new cutoff $U/2$ is at most~$s$, and the exploration has already processed everything below it. Every ball that is still growing therefore contains all centers it needs and can be terminated without being hit. Balls that were already full continue in later rounds under the smaller cutoff. Terminated balls receive no new entries, but their entries are still propagated. We prove this in \Cref{lem:batched-growth} (see \Cref{fig:termination}).

\begin{figure}[tb]
\centering
\begin{tikzpicture}[
  x=1cm,y=1cm,font=\small,
  vertex/.style={circle,draw=black!70,fill=white,inner sep=0pt,minimum size=5pt},
  retained/.style={vertex,double,double distance=1pt},
  card/.style={draw=black!60,rounded corners=1pt,minimum width=.5cm,minimum height=.34cm,
    inner sep=1pt,font=\scriptsize},
  stored/.style={card,fill=blue!12,draw=blue!65!black},
  missing/.style={card,dashed,draw=black!45,text=black!45},
  tree/.style={draw=black!70,line width=.7pt},
  chord/.style={->,>=stealth,draw=orange!80!black,dashed,line width=1.1pt},
  event/.style={->,>=stealth,draw=blue!65!black,line width=1.3pt},
  annotation/.style={align=center,font=\scriptsize}
]
\begin{scope}
  \node[font=\footnotesize\bfseries,anchor=west] at (-.4,1.2) {(a)};
  \coordinate (A) at (0,0); \coordinate (U) at (2.0,0);
  \coordinate (X) at (4.45,0); \coordinate (V) at (6.4,0);
  \draw[black!35,dotted,line width=.9pt] (A)--(U) (X)--(V);
  \draw[event] (U)--(X) node[midway,above=1pt,annotation,text=blue!65!black] {pending event\\$(a,u,x)$};
  \node[vertex,label=above:{$a$}] at (A) {};
  \node[retained,label=above:{$u$}] at (U) {};
  \node[vertex,label=above:{$x$}] at (X) {};
  \node[vertex,label=above:{$v$}] at (V) {};
  \node[stored] at (0,-.42) {$a$};
  \node[stored] at (2.0,-.42) {$a$};
  \node[missing] at (4.45,-.42) {$a$};
  \node[missing] at (6.4,-.42) {$a$};
  \node[annotation,anchor=north] at (0,-.7) {root};
  \node[annotation,anchor=north] at (2.0,-.7) {may be terminated;\\entries retained};
  \node[annotation,anchor=north] at (4.45,-.7) {first vertex\\without~$a$: active};
  \node[annotation,anchor=north] at (6.4,-.7) {active,\\misses~$a$};
\end{scope}
\begin{scope}[xshift=8.6cm]
  \node[font=\footnotesize\bfseries,anchor=west] at (-1.1,1.2) {(b)};
  \coordinate (a) at (0,1.0); \coordinate (w) at (0,.25);
  \coordinate (u) at (-.7,-.6); \coordinate (v) at (.7,-.6);
  \draw[tree] (a)--node[left,annotation]{2}(w) (w)--node[left,annotation]{1}(u) (w)--node[right,annotation]{2}(v);
  \draw[chord] (u)--node[below,annotation]{2}(v);
  \node[vertex,label=right:{$a$}] at (a) {};
  \node[vertex,label=right:{$w$}] at (w) {};
  \node[vertex,label=left:{$u$}] at (u) {};
  \node[vertex,label=right:{$v$}] at (v) {};
  \begin{scope}[xshift=1.3cm,yshift=.45cm,x=.5cm]
    \fill[blue!12] (0,-.11) rectangle (4.5,.11);
    \draw[->,>=stealth,black!70] (0,0)--(5.8,0);
    \foreach \t in {0,4.5,5} \draw[black!70] (\t,-.1)--(\t,.1);
    \node[annotation,below=2pt] at (0,0) {$0$};
    \node[annotation,below=2pt,text=blue!65!black] at (4.1,-.05) {$\rho_{\rm new}=\tfrac92$};
    \node[annotation,above=2pt,text=orange!80!black] at (5,.05) {$s=5$};
    \node[annotation,above=2pt,text=blue!65!black] at (2.0,.05) {already stored};
    \node[annotation,anchor=west] at (0,-.85) {$D[a,u]=3$, $D[a,v]=4$};
    \node[annotation,anchor=west] at (0,-1.25) {$U_{\rm new}=9$, $\ell(Z)=5$};
  \end{scope}
\end{scope}
\end{tikzpicture}
\caption{Illustration of safe termination. (a)~The first vertex $x\in P(a,v)$ that does not store~$a$ is active, so the event $(a,u,x)$ is pending. (b)~A cycle closed at $s=5$ gives $U_{\rm new}=9$, so the new cutoff $9/2$ is at most~$s$.}
\label{fig:termination}
\end{figure}
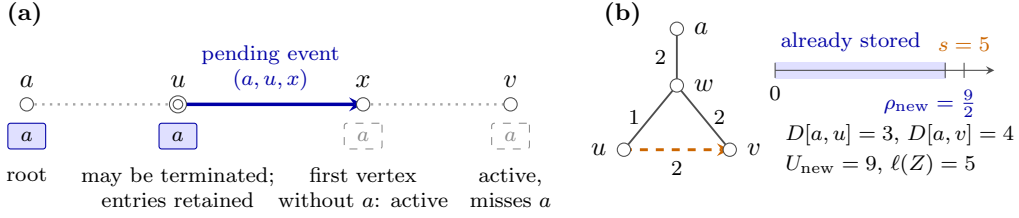

\paragraph{Completing the algorithm.}
We construct the levels one by one, sharing the best certificate among them; during the construction, $U$ and the cutoff $U/2$ may decrease. We then fix the final value~$U_0$ and run the cluster searches and edge-and-pivot checks of~\cite{kadria2023improved} on the resulting hierarchy, with every cluster search truncated at $U_0/2$. Each level takes $O((m+n^{1+1/\kappa})\log n)$ time, which gives the running time $O((m+n^{1+1/\kappa})\kappa\log n)$.

\paragraph{Stretch $2k/3$.}
The algorithm of~\cite{kadria2026tighter} for odd~$k$ differs in its first level, which is sampled with probability $n^{-1/k}$ instead of $n^{-2/k}$ and whose sets $B_0(u)$ are computed explicitly. Since these sets are computed in full anyway, the first level can be chosen by Yan's schedule directly, without cycle detection; the remaining levels are constructed as above (see \Cref{sec:odd} for more details).

\section{Deterministic girth approximation algorithm}
\label{sec:deterministic}

In this section we prove \Cref{thm:deterministic-girth}, and in \Cref{sec:odd} we prove \Cref{thm:deterministic-2k3}.

\subsection{Truncating the cluster searches}
\label{sec:truncated-clusters}

For $A'\subseteq A\subseteq V$ and $\rho\in(0,+\infty]$, define the \emph{truncated bunch} of a vertex~$v$ and the \emph{truncated cluster} of a center~$a$ as the vertex sets
\begin{align}
 B_{A,A'}^{<\rho}(v)
 &:=\{a\in A:\delta(a,v)<\min\{\delta(v,A'),\rho\}\},\label{eq:truncated-bunch}\\
 \Cl_{A,A'}^{<\rho}(a)
 &:=\{v\in V:\delta(a,v)<\min\{\delta(v,A'),\rho\}\}.
\label{eq:truncated-cluster}
\end{align}
For $A=A_i$, $A'=A_{i+1}$, and $a\in A_i\setminus A_{i+1}$, the set $\Cl_{A_i,A_{i+1}}^{<\rho}(a)$ is the vertex set of $\Cl(a)\cap G_{<\rho}(a)$, illustrated in \Cref{fig:cluster-truncation}. Double counting gives
\begin{equation}\label{eq:truncated-duality}
 \sum_{a\in A\setminus A'}|\Cl_{A,A'}^{<\rho}(a)|
 =\sum_{v\in V}|B_{A,A'}^{<\rho}(v)|.
\end{equation}

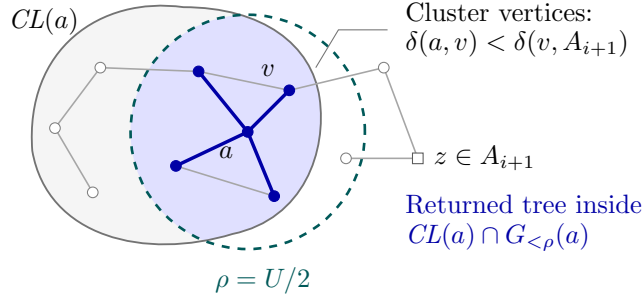
\begin{figure}[tb]
\centering
\begin{tikzpicture}[
  x=1cm,y=1cm,font=\small,
  vertex/.style={circle,draw=black!50,fill=white,inner sep=0pt,minimum size=4pt},
  retained vertex/.style={vertex,draw=blue!65!black,fill=blue!65!black},
  graph edge/.style={draw=black!35,line width=.6pt},
  retained edge/.style={draw=blue!65!black,line width=1.3pt}
]
\def\clusterboundary{(-2.85,-.25)
  .. controls (-2.95,1.1) and (-1.95,1.8) .. (-.85,1.65)
  .. controls (.2,1.65) and (.85,1.15) .. (.95,.45)
  .. controls (1.05,-.25) and (.75,-1.15) .. (-.2,-1.35)
  .. controls (-1.35,-1.7) and (-2.65,-1.4) .. cycle}
\fill[black!4] \clusterboundary;
\begin{scope}
  \clip (0,0) circle (1.55);
  \fill[blue!12] \clusterboundary;
\end{scope}
\draw[black!55,line width=.7pt] \clusterboundary;
\draw[teal!70!black,dashed,line width=1pt] (0,0) circle (1.55);

\coordinate (a) at (0,0);
\coordinate (p) at (-.65,.8);
\coordinate (q) at (-.95,-.45);
\coordinate (v) at (.55,.55);
\coordinate (r) at (.35,-.85);
\coordinate (s) at (-1.95,.85);
\coordinate (t) at (-2.05,-.8);
\coordinate (z) at (-2.55,.05);
\coordinate (x) at (1.3,-.35);
\coordinate (b) at (2.25,-.35);
\coordinate (y) at (1.8,.85);
\draw[graph edge] (p)--(s)--(z)--(t)
  (v)--(y)--(b)--(x) (p)--(v) (q)--(r);
\draw[retained edge] (a)--(p) (a)--(q) (a)--(v) (a)--(r);
\foreach \point in {s,t,z,x,y}
  \node[vertex] at (\point) {};
\foreach \point in {p,q,r}
  \node[retained vertex] at (\point) {};
\node[retained vertex,label=below left:{$a$}] at (a) {};
\node[retained vertex,label=above left:{$v$}] at (v) {};
\node[rectangle,draw=black!65,fill=white,inner sep=2pt,
  label=right:{$z\in A_{i+1}$}] at (b) {};

\node[anchor=south] at (-2.7,1.15) {$\Cl(a)$};
\node[anchor=west,align=left] at (1.95,1.35)
  {Cluster vertices:\\[-1pt]$\delta(a,v)<\delta(v,A_{i+1})$};
\draw[black!55] (.9,.85)--(1.2,1.35)--(1.85,1.35);
\node[anchor=north,text=teal!70!black] at (.2,-1.65)
  {$\rho=U/2$};
\node[anchor=west,text=blue!65!black,align=left] at (1.95,-1.2)
  {Returned tree inside\\[-1pt]$\Cl(a)\cap G_{<\rho}(a)$};
\end{tikzpicture}
\caption{Illustration of the truncated cluster of~$a$: the tree returned by the search (thick) and the cutoff $\rho=U/2$ (dashed).}
\label{fig:cluster-truncation}
\end{figure}

Let $\ClusterOrCycleBounded(a,\rho)$ denote $\ClusterOrCycle(a)$ stopped before extracting an edge with key at least~$\rho$. The \emph{stopping value} $\tau\in(0,\rho]$ of the bounded run is the key of the cycle-closing edge if it reports a cycle, and $\rho$ otherwise (including when the cluster is exhausted).

\begin{lemma}[Bounded cluster search]\label{lem:bounded-search}
Let $a\in A_i\setminus A_{i+1}$, $\rho\in(0,+\infty]$, and let $\tau$ be the stopping value of $\ClusterOrCycleBounded(a,\rho)$. Then:
\begin{enumerate}
 \item[(B1)] If a cycle is reported, its certificate value is at most~$2\tau$.
 \item[(B2)] The tree returned by the search contains, as an edge, every cluster orientation $(v,w)$ with $\delta(a,v)+\ell(v,w)<\tau$, and hence a shortest path to every $x\in\Cl_V(a)$ with $\delta(a,x)<\tau$.
\end{enumerate}
Moreover, the search takes
\begin{equation}\label{eq:bounded-cluster-time}
 O\!\left((1+|\Cl_{A_i,A_{i+1}}^{<\rho}(a)|)\log n\right)
\end{equation}
time.
\end{lemma}
\begin{proof}
If the edge $(v,w)$ closes a cycle, then $D[a,v]+\ell(v,w)=\tau$ and $w$ was discovered earlier, at a key of at most~$\tau$, so the certificate value $D[a,v]+\ell(v,w)+D[a,w]$ is at most~$2\tau$. This proves (B1).

An extracted orientation with key less than~$\tau$ does not close a cycle, so it is a tree edge or the reverse of one. Suppose that some cluster orientation with $\delta(a,v)+\ell(v,w)<\tau$ is never extracted, and choose one, $(v,w)$, with minimum key. Vertex~$v$ is discovered: either $v=a$, or the last edge of $P(a,v)$ is a cluster orientation with the smaller key $\delta(a,v)$, which is therefore extracted. By \Cref{lem:cluster-or-cycle}, vertex~$v$ inserts its cluster orientations in order of length, and those up to $(v,w)$ have keys at most $\delta(a,v)+\ell(v,w)<\tau$. The search stops only when it extracts an orientation of key~$\tau$, when the heap minimum reaches $\rho\geq\tau$, or when the heap is empty. Hence every orientation of~$v$ up to $(v,w)$ is inserted and then extracted in turn, a contradiction. The second part of (B2) follows by the argument before \Cref{lem:krsvz-analysis}.

For the running time, an edge $e=\{v,w\}$ extracted in orientation $(v,w)$ has key $D[a,v]+\ell(e)=\delta(a,v)+\ell(e)$ by \Cref{lem:cluster-or-cycle}, and $\max\{\delta(a,e),\delta(a,w)\}\leq\delta(a,v)+\ell(e)<\rho$. Hence every extracted edge and both of its endpoints lie in $\Cl(a)\cap G_{<\rho}(a)$. Before stopping, only tree edges, at most one reverse tree edge per discovered vertex, and possibly one cycle-closing edge are extracted, and the heap holds at most one pending orientation per discovered vertex. Since every discovered vertex lies in $\Cl_{A_i,A_{i+1}}^{<\rho}(a)$, the bound~\eqref{eq:bounded-cluster-time} follows.
\end{proof}

For a radius~$\Delta$, we call the (bounded or unbounded) cluster search for $a\in A_i\setminus A_{i+1}$ \emph{$\Delta$-good} if it satisfies (C1) or (C2) with $\Delta$ in place of~$\Delta_i$. 

\begin{corollary}\label{cor:bounded-good}
If $\Delta<\rho$, then $\ClusterOrCycleBounded(a,\rho)$ is $\Delta$-good.
\end{corollary}
\begin{proof}
Let $\tau$ be the stopping value. If $\tau\leq\Delta$, then $\tau<\rho$, so the search reported a cycle, and by~(B1) its certificate has value at most $2\tau\leq2\Delta$; this is~(C1). Otherwise $\tau>\Delta$, and by~(B2) the returned tree contains every cluster orientation $(v,w)$ with $\delta(a,v)+\ell(v,w)\leq\Delta<\tau$; this is~(C2).
\end{proof}

\paragraph{The bounded algorithm.}
Given a hierarchy $V=A_0\supseteq\cdots\supseteq A_\kappa=\emptyset$ and a cycle certificate of value~$U_0$ (where $U_0=+\infty$, with no certificate, is allowed), let $\BoundedCycle$ denote algorithm $\Cycle$ with the following changes: 
\begin{itemize}
 \item the initialization phase of $\Cycle$ uses the given sets instead of sampling them, and the best certificate is initialized to the given one;
 \item every cluster search is $\ClusterOrCycleBounded(a,U_0/2)$, with the cutoff $U_0/2$ fixed even when better certificates are found, and the edge-and-pivot checks use the returned trees;
 \item all hash tables are replaced by balanced search trees.
\end{itemize}

\begin{lemma}[Truncation transfer]\label{lem:truncation-transfer}
For every hierarchy $V=A_0\supseteq\cdots\supseteq A_\kappa=\emptyset$ and every cycle certificate of value~$U_0$ (or $U_0=+\infty$), algorithm $\BoundedCycle$ returns a cycle of length at most~$4\kappa g/3$.
\end{lemma}
\begin{proof}
If $U_0\leq4\kappa g/3$, the given certificate suffices. Otherwise $\rho=U_0/2>2\kappa g/3\geq\Delta_i$ for every $i<\kappa$, by~\eqref{eq:analysis-radii}. Hence, by \Cref{cor:bounded-good}, every bounded search is good, and \Cref{lem:krsvz-analysis} yields a certificate of value at most $4\kappa g/3$.
\end{proof}

\subsection{Constructing one level}
\label{sec:growable-prefixes}

Fix the current level $A\subseteq V$ with $N=|A|$. We must choose the next level $R\subseteq A$ so that $R$ is small and the truncated bunches $B_{A,R}^{<U_{\rm out}/2}(v)$ have small total size, where $U_{\rm out}$ is the certificate value at the end of the construction of the level.

\paragraph{Balls.}
Order the centers of~$A$ at a vertex~$v$ by the key
\[
 \theta_v(a):=(\delta(a,v),\operatorname{id}(a)),
\]
and let $\operatorname{rank}_v(a)$ be the position of~$a$ in this order. Every vertex~$v$ stores a list $L(v)$ of \emph{entries}, called its \emph{ball}. An entry records a center $a\in A$, a distance $D[a,v]$, and a pointer to a predecessor entry of~$a$ at a neighbor of~$v$. The ball is \emph{exact} if its centers are the first $|L(v)|$ centers in the order~$\theta_v$, in this order, and $D[a,v]=\delta(a,v)$ for each entry. Since ties are broken by identifiers, an exact ball may contain only some of the centers at its largest distance. We identify $L(v)$ with its set of centers when writing inclusions. Finally, the \emph{hit threshold} of~$v$ is
\[
 \sigma_R(v):=\min_{c\in R}\theta_v(c),
\]
with value $(+\infty,+\infty)$ if $R=\emptyset$.

\begin{restatable}[One-level construction]{lemma}{onelevellemma}\label{lem:one-level}
Let $A\subseteq V$, $N=|A|\geq1$, and let $1\leq r\leq N$ be an integer. Given a cycle certificate of value $U_{\rm in}$ (or $U_{\rm in}=+\infty$ with no certificate), a deterministic algorithm computes $R\subseteq A$ and a cycle certificate of value $U_{\rm out}\leq U_{\rm in}$ (none if $U_{\rm out}=+\infty$) such that
\begin{equation}\label{eq:one-level-bounds}
 |R|\leq r
 \qquad\text{and}\qquad
 \sum_{v\in V}|B_{A,R}^{<U_{\rm out}/2}(v)|=O(nN/r).
\end{equation}
Its running time is $O((m+nN/r)\log n)$.
\end{restatable}

The construction (\Cref{alg:one-level}) proceeds in \emph{rounds} $t=0,1,\ldots$. Round~$t$ consists of one \emph{growth call} with quota $b_t=\min\{N,2^tb_0\}$, followed by one \emph{selection phase}. We use a weighted variant of the schedule of Yan~\cite[Section~2, Algorithm~1]{yan2026lossless} with $p=1$, where the cost of a ball of size~$x$ is~$x^p$: quotas double, and each selection phase hits balls greedily by the weights of \Cref{alg:one-level} until at most one eighth of their total weight remains unhit. This makes both the number and the total degree of the unhit balls drop by a factor of four (\Cref{lem:selection}).

\begin{algorithm2e}[tbp]
\caption{One-level construction}\label{alg:one-level}
\KwIn{$A\subseteq V$, budget $r$, and a cycle certificate of value $U$ (possibly $+\infty$)}
$R\leftarrow\emptyset$; $L(v)\leftarrow\emptyset$ for all $v\in V$; mark all balls live\;
\lForEach{$a\in A$}{insert the root entry $(a,D[a,a]=0)$, with no predecessor, in $L(a)$}
$b_0\leftarrow\lceil8|A|/r\rceil$; $t\leftarrow0$\;
\While{a live ball remains}{
 $b_t\leftarrow\min\{|A|,2^t b_0\}$\;
 $\Grow(A,R,b_t,U,L)$ (\Cref{alg:growth}), updating $U$ and the ball states\;
 $F\leftarrow\{v:L(v)\text{ is quota-paused}\}$; $F_0\leftarrow F$\;
 $\omega(v)\leftarrow\deg(F_0)+|F_0|\deg(v)$ for all $v\in F$, where $\deg(F_0)=\sum_{v\in F_0}\deg(v)$; $W_0\leftarrow\sum_{v\in F}\omega(v)$\;
 \While{$8\sum_{v\in F}\omega(v)>W_0$}{
  Choose $c\in A\setminus R$ maximizing $\sum_{v\in F:\,c\in L(v)}\omega(v)$, breaking ties by smaller identifier\;
  $R\leftarrow R\cup\{c\}$\;
  Mark every $v\in F$ with $c\in L(v)$ hit-terminated and remove it from $F$\;
 }
 $t\leftarrow t+1$\;
}
\KwRet{$R$, and $U$ with its certificate if $U<+\infty$}
\end{algorithm2e}

The root entry of~$a$ is the first entry of $L(a)$, since $\theta_a(a)=(0,\operatorname{id}(a))$ and every other center has positive distance from~$a$. Root entries count toward the first quota. 

\paragraph{Ball states.}
A ball is \emph{live} until it is permanently terminated. During a growth call with quota~$b$, a live ball is \emph{active} if $|L(v)|<b$, and \emph{quota-paused} if $|L(v)|=b$. Only active balls receive new entries; a quota-paused ball becomes active again when the quota increases, unless selection terminates it first.

A ball is permanently terminated in one of two ways. It is \emph{hit-terminated} when selection chooses a center already in its stored ball, or when growth appends a selected center to it. It is \emph{cutoff-terminated} when growth stops at the current cutoff. \Cref{fig:prefix-states} shows all transitions. In every state, the stored entries of a ball are still propagated to other vertices.

\begin{figure}[tb]
\centering
\begin{tikzpicture}[
  x=1cm,y=1cm,font=\small,
  state/.style={draw=black!70,rounded corners=3pt,minimum width=2.6cm,
    minimum height=.8cm,align=center,fill=white},
  live state/.style={state,draw=blue!65!black,line width=.9pt},
  final state/.style={state,double,draw=black!70},
  move/.style={->,>=stealth,line width=.7pt,draw=black!75},
  growth/.style={move,draw=blue!65!black},
  selection/.style={move,draw=orange!80!black,dashed},
  annotation/.style={align=center,font=\footnotesize}
]
\draw[blue!65!black,dashed,rounded corners=6pt] (-1.75,1.25) rectangle (7.75,-1.45);
\node[anchor=north west,text=blue!65!black,font=\footnotesize] at (-1.7,1.2) {live};

\node[live state] (active) at (0,0) {active\\[-1pt]$|L(v)|<b$};
\node[live state] (paused) at (6,0) {quota-paused\\[-1pt]$|L(v)|=b$};
\node[final state] (cutoff) at (0,-3.4) {cutoff-terminated};
\node[final state] (hit) at (6,-3.4) {hit-terminated};

\draw[move] (-2.4,0) -- (active);
\draw[growth] (active) to[bend left=18]
  node[annotation,above] {growth: quota reached} (paused);
\draw[selection] (paused) to[bend left=18]
  node[annotation,below,pos=.3] {next round:\\quota doubles} (active);
\draw[growth] (active) --
  node[annotation,left,align=right,pos=.7] {growth stops: cycle is found, \\or queue minimum ${}\geq\rho$} (cutoff);
\draw[selection] (paused) --
  node[annotation,right,align=left,pos=.65] {selection:\\$L(v)$ contains\\the chosen center} (hit);
\draw[growth] (active.south east) to[out=-40,in=180]
  node[annotation,pos=.55,above right=-1pt] {growth: selected\\center arrives} (hit.west);
\end{tikzpicture}
\caption{Ball states. Solid (blue) transitions happen during growth calls, and dashed (orange) ones between them.}
\label{fig:prefix-states}
\end{figure}
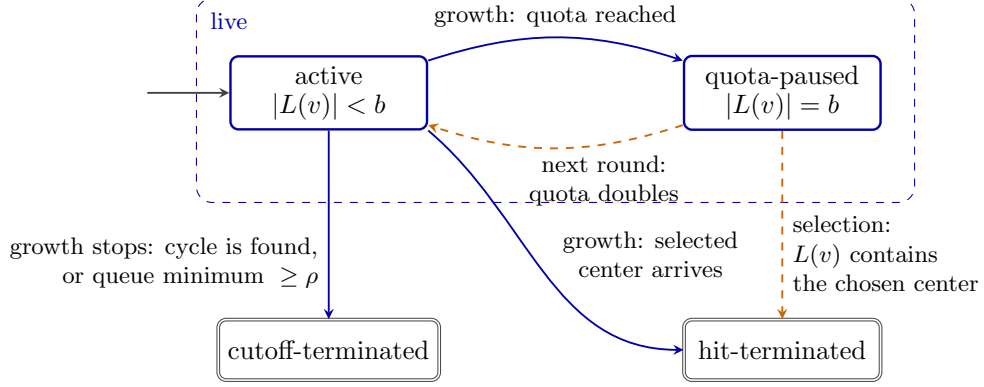

\paragraph{Properties of growth calls.}
\Cref{sec:batched-growth} describes $\Grow$ and proves (\Cref{lem:batched-growth}) that it has the following properties.
\begin{enumerate}
 \item[(G1)] Every ball is always exact.
 \item[(G2)] Growth appends a center~$a$ to $L(v)$ only if $\theta_v(a)\leq\sigma_R(v)$, and it hit-terminates~$v$ if $a\in R$. In particular, growth keeps every live ball disjoint from~$R$.
 \item[(G3)] Only active balls receive entries, and when a growth call returns, every live ball is quota-paused.
\end{enumerate}

\begin{restatable}[Selection]{lemma}{selectionlemma}\label{lem:selection}
Assume (G1)--(G3), and let $S_t$ be the set of vertices whose balls are live at the start of round~$t$. Then $|S_t|\leq n/4^t$ and $\sum_{v\in S_t}\deg(v)\leq2m/4^t$. Moreover, \Cref{alg:one-level} performs $O(\log n)$ rounds and returns a set~$R$ with $|R|\leq r$, and when it terminates, the total number of stored entries satisfies
\begin{equation}\label{eq:exposed-mass}
 X:=\sum_{v\in V}|L(v)|\leq2nb_0=O(nN/r),
\end{equation}
and all selection phases together take $O((nN/r)\log n)$ time.
\end{restatable}
\ifC
\else
\begin{proof}
We first show that every live ball satisfies $L(v)\cap R=\emptyset$. This holds initially since $R=\emptyset$. Growth preserves it by (G2), and selection preserves it since it terminates every live ball that contains the newly selected center. Therefore, by (G1), every ball in~$F$ consists of exactly $b_t$ centers of $A\setminus R$.

\paragraph{The number and the total degree of the live balls.}
Let $F_0$ be the set~$F$ when the weights are assigned, so that $\omega(v)=\deg(F_0)+|F_0|\deg(v)$, where $\deg(F_0)=\sum_{v\in F_0}\deg(v)$. Since $G$ is connected and contains a cycle, every vertex has degree at least~$1$, and therefore $\deg(F_0)\geq|F_0|\geq1$ whenever $F_0\neq\emptyset$. The initial weight is $W_0=\sum_{v\in F_0}\omega(v)=2|F_0|\deg(F_0)$. If $F_0=\emptyset$, the selection phase is vacuous and the bounds below hold trivially, so we assume that $F_0\neq\emptyset$, and hence $W_0>0$. 

Let $F'$ be the set of balls of~$F_0$ that are not hit when the selection phase ends, and let $W=\sum_{v\in F'}\omega(v)$. By the definition of the weights, $W=|F'|\deg(F_0)+|F_0|\deg(F')$, where $\deg(F')=\sum_{v\in F'}\deg(v)$, that is,
\[
 \frac{W}{W_0}=\frac12\left(\frac{|F'|}{|F_0|}+\frac{\deg(F')}{\deg(F_0)}\right).
\]
By the stopping rule, $W/W_0\leq1/8$. Both fractions in this identity are nonnegative, so each is at most $1/4$, that is, $|F'|\leq|F_0|/4$ and $\deg(F')\leq\deg(F_0)/4$. 

By (G3), every ball that is live after the growth call of round~$t$ is quota-paused. Hence, the balls that are live at the start of round $t+1$ are exactly the balls of~$F'$. Since $F_0\subseteq S_t$, it follows that $|S_{t+1}|\leq|S_t|/4$ and $\sum_{v\in S_{t+1}}\deg(v)\leq\frac14\sum_{v\in S_t}\deg(v)$. Since $S_0\subseteq V$, we get that $|S_t|\leq n/4^t$ and $\sum_{v\in S_t}\deg(v)\leq2m/4^t$, as required.

\paragraph{The number of selected centers.}
Consider the selection phase of round~$t$ with $b_t<N$, and let $W$ be the total weight of the balls that are not hit yet. If we sum, over all centers of $A\setminus R$, the weight of these balls that contain the center, we get $b_tW$. Therefore, some center has weight at least $b_tW/N$, and each greedy selection multiplies~$W$ by at most $1-b_t/N$. After $j$ selections, we have $W\leq W_0(1-b_t/N)^j<W_0\,e^{-jb_t/N}$. Thus, the phase ends after at most $\lceil N\ln8/b_t\rceil<4N/b_t$ selections, where the inequality follows from $\ln8+1<4$ and $N/b_t>1$. This adapts the $p=1$ analysis of Yan~\cite{yan2026lossless} to weighted balls. Since $b_t\geq2^tb_0\geq2^{t+3}N/r$, round~$t$ makes fewer than $r/2^{t+1}$ selections.

Next, we consider the case that $b_t=N$. In this case, a quota-paused ball contains all of~$A$, so it can exist only if $R=\emptyset$, and then one selection hits all of~$F$. Hence, round~$t$ is the last round whenever $b_t=N$, and there are at most $1+\max\{0,\lceil\log_2(N/b_0)\rceil\}=O(\log n)$ rounds. If every round has $b_t<N$, the rounds make fewer than $\sum_{t\geq0}r/2^{t+1}=r$ selections. Otherwise, the last round~$t$ has $b_t=N$ and makes at most one selection. If $t\geq1$, then round $t-1$ had $b_{t-1}=2^{t-1}b_0<N$, and since $b_0\geq8N/r$, we get $2^{t-1}\cdot8N/r<N$, that is, $r/2^t>4$. Moreover, the previous rounds made fewer than $r(1-2^{-t})$ selections, so $|R|<r(1-2^{-t})+1<r$. If $t=0$, we use $r\geq1$. Therefore, $|R|\leq r$, as required.

\paragraph{The number of entries and the selection time.}
Terminated balls keep their entries, which are counted in~$X$. At most $|S_t|\leq n/4^t$ balls are live at the start of round~$t$, and each of them stores at most $b_t\leq2^t b_0$ entries when it is terminated in that round. By charging every ball to its termination round, we get that
\[
 X\leq\sum_{t\geq0}(n/4^t)\,2^t b_0=2nb_0=O(nN/r),
\]
where the last equality follows from $b_0=\lceil8N/r\rceil\leq8N/r+1\leq9N/r$, since $r\leq N$. 

The selection phase of round~$t$ uses the incidence lists of the at most $n/4^t$ balls in~$F$ and a max-heap that stores, for every center occurring in them, the total weight of the balls that contain it. When a ball is hit, we subtract its weight from all the centers it contains. The weights are integers of value $O(n^2)$. Each of the at most $(n/4^t)\,b_t$ incidences is inserted and removed at most once in that round, so the round costs $O((n/4^t)\,2^tb_0\log n)$ time. By the same geometric sum, all selection phases together cost $O((nN/r)\log n)$ time, as required.
\end{proof}

\fi

\subsection{Growing the balls}
\label{sec:batched-growth}

We describe a growth call (\Cref{alg:growth}) using a conceptual event queue; \Cref{sec:implementation} implements this queue without materializing an event for every entry-edge pair.

During one growth call, $R$ and~$b=b_t$ are fixed, while $\rho=U/2$ may decrease. For every stored entry of~$a$ at~$u$ and every oriented edge $(u,v)$ into an active vertex, the event queue contains the \emph{propagation event} $(a,u,v)$, with \emph{tail}~$u$ and \emph{head}~$v$, unless it has already been processed in this or an earlier call of the current level. Its key is
\begin{equation}\label{eq:propagation-key}
 K(a,u,v):=(D[a,u]+\ell(u,v),\operatorname{id}(a)).
\end{equation}
Keys are compared lexicographically, like the keys~$\theta_v$; the first coordinate of a key is its \emph{distance coordinate}. Events are processed in increasing order of key, and ties are broken by the identifiers of the head and the tail. \footnote{Formally, an event has priority $(K(a,u,v),\operatorname{id}(v),\operatorname{id}(u))$, compared lexicographically. }

Each entry-edge orientation is processed at most once in the whole level. Removing the events into a vertex that becomes inactive does not mark them processed, and if the vertex resumes in a later call, they return to the queue. Notice that the algorithm does not compute the hit thresholds. Instead, a vertex is hit-terminated when a selected center arrives at it.

\begin{algorithm2e}[htbp]
\caption{$\Grow(A,R,b,U,L)$}\label{alg:growth}
Mark each live ball active if $|L(v)|<b$ and quota-paused otherwise\;
$Q\leftarrow$ the unprocessed propagation events $(a,u,v)$ for every stored entry of~$a$ at~$u$ and every edge $(u,v)$ with $v$ active\;
\While{$Q\neq\emptyset$ and the minimum event of $Q$ has distance coordinate less than $U/2$}{
 Extract the minimum event $(a,u,v)$, with key $(s,\operatorname{id}(a))$, from $Q$ and mark it processed\;
 \uIf{$a\notin L(v)$}{
  Append $a$ to $L(v)$ with $D[a,v]=s$ and predecessor the entry of $a$ at $u$\;
  Add the events $(a,v,w)$ for all edges $(v,w)$ with $w$ active\;
  \uIf{$a\in R$}{mark $v$ hit-terminated and remove all events into $v$\;}
  \lElseIf{$|L(v)|=b$}{mark $v$ quota-paused and remove all events into $v$}
 }
 \ElseIf{$\{u,v\}$ is not a tree edge of $T_a$}{
  $U\leftarrow D[a,u]+\ell(u,v)+D[a,v]$, and record the cycle~$Z$ of~\eqref{eq:fundamental-cycle} as its certificate\;
  \textbf{break}\;
 }
}
Mark every still-active ball cutoff-terminated\;
\end{algorithm2e}

The predecessor pointers of the entries of a center~$a$ are never changed, and form a tree~$T_a$. When $U$ is updated, the cycle~$Z$ is extracted immediately from the predecessor chains as in \Cref{sec:preliminaries}, in $O(n)$ time, so that it remains available after the trees of the level are discarded. 

\begin{restatable}[Correctness of growth]{lemma}{growthlemma}\label{lem:batched-growth}
Throughout \Cref{alg:one-level}, (G1) holds, and for every center~$a$ the tree~$T_a$ is rooted at~$a$, spans the vertices that store~$a$, and its root paths are shortest paths.  Growth calls satisfy (G2) and (G3). Every update of~$U$ strictly decreases it and records a cycle of length at most the new value. When \Cref{alg:one-level} terminates,
\begin{equation}\label{eq:prefix-coverage}
 B_{A,R}^{<U/2}(v)\subseteq L(v)
 \qquad(v\in V).
\end{equation}
\end{restatable}
\ifC
\else
\begin{proof}
Property (G3) holds by construction, since only active balls receive entries and the last step of \Cref{alg:growth} cutoff-terminates all balls that are still active. We prove the remaining claims by induction on the steps of \Cref{alg:one-level}. At each step, let
\[
 \Lambda(v):=\{a\in A:\theta_v(a)\leq\sigma_R(v),\ \delta(a,v)<\rho\}
\]
be the set of centers \emph{required} at~$v$, that is, the centers that lie below the cutoff and do not come after the nearest selected center at~$v$. Notice that, because of the non-strict inequality, the nearest selected center is required whenever it lies below the cutoff.

Adding centers to~$R$ lowers $\sigma_R(v)$, and $\rho$ never increases, so $\Lambda(v)$ can only shrink over time. Moreover, $\Lambda(v)$ is a prefix of the order~$\theta_v$. Consequently, if $L(v)$ is an exact ball that misses a member of~$\Lambda(v)$, then the first center missing from~$L(v)$ also belongs to~$\Lambda(v)$. We maintain the following three invariants:
\begin{enumerate}
 \item[(I)] Every $L(v)$ is an exact ball, and the predecessor path of each entry of~$a$ at~$v$ is a shortest path from~$a$ to~$v$.
 \item[(II)] Every permanently terminated vertex satisfies $\Lambda(v)\subseteq L(v)$.
 \item[(III)] Every live ball is disjoint from~$R$.
\end{enumerate}
The invariants hold initially, after the root entries are inserted. Since $\Lambda(v)$ only shrinks and terminated balls receive no entries, a terminated vertex continues to satisfy~(II) once it does. We first show a queue property.

\paragraph{No required entry is hidden.}
Let $x\in P(a,v)$. By the triangle inequality and since identifiers break ties in the same way at~$x$ and at~$v$, every center that strictly precedes~$a$ at~$x$ also strictly precedes~$a$ at~$v$. Consequently,
\begin{equation}\label{eq:rank-monotonicity}
 \operatorname{rank}_x(a)\leq\operatorname{rank}_v(a),
 \qquad a\in \Lambda(v)\ \Longrightarrow\ a\in \Lambda(x).
\end{equation}
For the second implication, notice that $\delta(a,x)\leq\delta(a,v)<\rho$, and that a selected center that strictly precedes~$a$ at~$x$ would strictly precede~$a$ at~$v$, which contradicts $\theta_v(a)\leq\sigma_R(v)$.

Suppose that $v$ is active and that $a\in \Lambda(v)\setminus L(v)$ has rank at most~$b$. Let $x$ be the first vertex of $P(a,v)$ that does not store~$a$. Since the root stores~$a$, the vertex~$u$ preceding~$x$ on $P(a,v)$ stores~$a$. By (II) and~\eqref{eq:rank-monotonicity}, $x$ is not permanently terminated. Moreover, $x$ is not quota-paused, since by (I) and since the rank of~$a$ at~$x$ is at most~$b$, a quota-paused ball at~$x$ would contain~$a$. Hence, $x$ is active.

Next, we show that the event $(a,u,x)$ is pending. The event has not been processed in this call or in any earlier call of the current level, since processing it would have left~$a$ stored at~$x$, either by inserting it or because it was already there, and entries are never deleted. By (I), its key is $\theta_x(a)\leq\theta_v(a)$. We have therefore proved the following queue property:
\begin{equation}\label{eq:queue-completeness}
 \min\operatorname{key}(\text{queue})\leq\theta_v(a)
 \quad\text{if }v\text{ is active},\
 a\in \Lambda(v)\setminus L(v),\
 \operatorname{rank}_v(a)\leq b.
\end{equation}
An empty queue has minimum $(+\infty,+\infty)$. The same argument applies at the start of a call, since the queue then contains every unprocessed event into an active vertex.

\paragraph{Accepting an entry, pausing, and hits.}
Consider a minimum propagation event $(a,u,v)$ with key $K=(s,\operatorname{id}(a))$, where $a\notin L(v)$ and $s<\rho$, and let $c$ be the first center missing from $L(v)$. By (I) and (III), no selected center strictly precedes~$c$ at~$v$, and therefore $\theta_v(c)\leq\sigma_R(v)$. By (I) and the path through~$u$, we also have $\theta_v(c)\leq\theta_v(a)\leq K$, and therefore $\delta(c,v)\leq s<\rho$. Thus, $c\in \Lambda(v)$ and $\operatorname{rank}_v(c)=|L(v)|+1\leq b$. If $\theta_v(c)<K$, then \eqref{eq:queue-completeness} contradicts the minimality of~$K$. Hence, $\theta_v(c)=\theta_v(a)=K$, which implies that $c=a$, that $s=\delta(a,v)$, and that $\theta_v(a)\leq\sigma_R(v)$. This proves (G2).

Appending the entry preserves (I), and attaching it to its already stored predecessor extends~$T_a$ by one vertex. If $a\notin R$, then $\theta_v(a)<\sigma_R(v)$, so (III) is preserved, and pausing at quota changes no invariant. If $a\in R$, then $\theta_v(a)=\sigma_R(v)$, so all the centers required at~$v$ are now stored, and hit-terminating~$v$ preserves (II) and~(III). Finally, processing a tree-edge duplicate changes no invariant.

\paragraph{Cutoff-termination and cycle updates.}
If the minimum event has distance coordinate at least~$\rho$, then any active vertex that misses a member of $\Lambda(v)$ again has a first missing center of rank at most~$b$. The distance of this center is less than~$\rho$, which contradicts~\eqref{eq:queue-completeness}. The same argument applies to an empty queue. Thus, all active vertices can be cutoff-terminated.

Next, consider a non-tree duplicate $(a,u,v)$ with distance coordinate~$s$, and let $w=\operatorname{LCA}_{T_a}(u,v)$. By (I) and the triangle inequality, $D[a,v]\leq D[a,u]+\ell(u,v)=s$. By~\eqref{eq:fundamental-cycle}, the extracted cycle satisfies
\begin{align*}
 \ell(Z)
 &=D[a,u]+\ell(u,v)+D[a,v]-2D[a,w]\\
 &\leq D[a,u]+\ell(u,v)+D[a,v]=U_{\rm new}
 \leq2s<U_{\rm old}.
\end{align*}
Therefore, the new cutoff satisfies $\rho_{\rm new}\leq s$. Immediately before this minimum event is processed, \eqref{eq:queue-completeness} excludes any missing required center of distance less than~$\rho_{\rm new}$ at an active vertex. Hence, cutoff-termination after the update also preserves (II).

\paragraph{Selection and restarting.}
A ball that is terminated by selection contains the newly selected center~$c$. By (III), it contains no earlier selected center, so $c$ is the nearest selected center at~$v$. By (I), the ball contains every center up to~$c$, and therefore it satisfies (II). Since selection terminates every live ball that contains~$c$, (III) is preserved. Increasing the quota changes no stored information. At the next call, the queue consists of the events that remain unprocessed, including the events into quota-paused vertices that resume, and the argument above reestablishes~\eqref{eq:queue-completeness}. Notice that an event that was removed earlier, because its head became quota-paused, may now be a non-tree duplicate. In this case it triggers a new cycle update, which is covered by the previous case and again strictly decreases~$U$. Thus, (I)--(III) hold across all calls.

This completes the induction. At termination, every vertex is permanently terminated, so (II) holds for every vertex. Since $B_{A,R}^{<U/2}(v)\subseteq\Lambda(v)$, this proves~\eqref{eq:prefix-coverage}, as required.
\end{proof}

\fi

\subsection{Implementation and running time}
\label{sec:implementation}

\begin{figure}[tb]
\centering
\begin{tikzpicture}[
  x=1cm,y=1cm,font=\small,
  entry/.style={draw=blue!65!black,fill=blue!12,rounded corners=1pt,minimum width=1.5cm,minimum height=.42cm,font=\scriptsize},
  absent/.style={entry,dashed,fill=white,draw=black!45,text=black!50},
  bucket/.style={draw=black!60,rounded corners=2pt,minimum width=3.1cm,minimum height=.42cm},
  edge dot/.style={circle,fill=black!45,inner sep=0pt,minimum size=3.5pt},
  min dot/.style={circle,fill=orange!80!black,inner sep=0pt,minimum size=5pt},
  link/.style={draw=black!35},
  rep/.style={->,>=stealth,draw=orange!80!black,line width=.8pt},
  annotation/.style={align=center,font=\scriptsize}
]
\node[annotation] at (0,.55) {stored ball $L(u)$};
\node[annotation] at (3.6,.55) {buckets of edges into active heads};
\node[annotation] at (7.4,.55) {global queue};
\foreach \j/\y in {1/0,2/-.6,3/-1.2}{
  \node[entry] (e\j) at (0,\y) {entry $\j$};
  \node[bucket] (h\j) at (3.6,\y) {};
  \node[annotation,anchor=west] at (1.2,\y) {$H_{u,\j}$};
  \draw[link] (e\j.east)--(1.15,\y);
}
\node[absent] (e4) at (0,-1.8) {entry 4 absent};
\node[bucket,dashed,draw=black!45] (h4) at (3.6,-1.8) {};
\node[annotation,anchor=west] at (1.2,-1.8) {$H_{u,4}$};
\node[annotation,text=black!55] at (3.6,-1.8) {waiting};
\foreach \y/\cnt in {0/4,-.6/2,-1.2/3}{
  \node[min dot] at (2.35,\y) {};
  \foreach \i in {2,...,\cnt} \node[edge dot] at (2.35+.55*\i-.55,\y) {};
}
\node[draw=black!60,rounded corners=2pt,minimum width=1.7cm,minimum height=1.7cm] (q) at (7.4,-.6) {};
\foreach \y in {0,-.6,-1.2} \draw[rep] (5.2,\y) .. controls (5.9,\y) .. (6.55,-.6+.35*\y+.2);
\node[annotation] at (7.4,-1.9) {one representative\\per bucket with an entry};
\begin{scope}[yshift=-2.95cm]
  \node[annotation] (s1) at (0,0) {head active\\$j(e)=j$};
  \node[annotation] (s2) at (3.6,0) {head paused: $e$ removed,\\$j(e)=j$ kept};
  \node[annotation] (s3) at (7.3,0) {head resumes: $e$ reinserted\\into $H_{u,j(e)}$};
  \draw[->,>=stealth,black!60] (s1)--(s2);
  \draw[->,>=stealth,black!60] (s2)--(s3);
  \node[annotation] at (3.6,-.75) {only processing the event of entry $j$ along~$e$ sets $j(e)\leftarrow j+1$};
\end{scope}
\end{tikzpicture}
\caption{Illustration of the buckets $H_{u,j}$ of~$u$. Only the minimum edge of a bucket is represented, and only if entry~$j$ of $L(u)$ exists.}
\label{fig:buckets}
\end{figure}
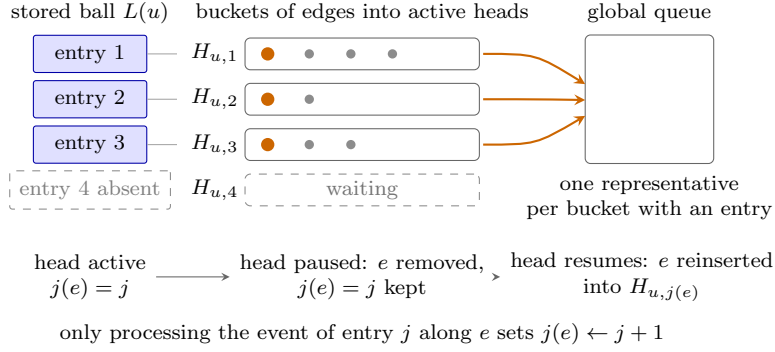

\ifC
\else
Constructing all events of \Cref{alg:growth} explicitly would cost $\sum_u\deg(u)|L(u)|$ time. We now describe an implementation that avoids this cost. Write $L(u)=(a_1(u),a_2(u),\ldots)$. For every oriented edge $e=(u,v)$, we maintain the index $j(e)$ of the first entry of $L(u)$ that has not been processed along~$e$ yet. The index is set to~$1$ when the construction of the level starts, and it is never reset. It advances only when the event of entry $j(e)$ along~$e$ is processed. Since the events along an edge are ordered by~$j$, only the first unprocessed event of the edge needs to be represented.

We group the edges with the same tail and the same index into buckets
\[
 H_{u,j}:=\{e=(u,v):v\text{ is active},\ j(e)=j\}.
\]
We maintain each bucket as an indexed min-heap ordered by $(\ell(u,v),\operatorname{id}(v))$. If entry~$j$ of~$L(u)$ exists, we represent only the minimum edge of $H_{u,j}$ in the global queue, with the priority of its propagation event. Otherwise, the bucket waits for entry~$j$. Since later entries on an edge have no smaller keys, these representatives give exactly the minimum event of \Cref{alg:growth}. A new entry at~$u$ activates at most one waiting bucket, the one for its index, without scanning the adjacency list of~$u$.

After processing $e=(u,v)$ from $H_{u,j}$, we remove it from that bucket and refresh or delete the old representative. If $v$ remains active, we insert $e$ into $H_{u,j+1}$. If entry~$j+1$ exists, we insert or decrease the representative of that bucket as needed, and otherwise the bucket waits for that entry. Thus, processing an edge updates at most two bucket representatives (see \Cref{fig:buckets}).

A growth call starts by scanning the adjacency lists of the active vertices and inserting their incoming edges into the buckets given by their saved indices. When a head becomes inactive, we remove its incoming edges using heap handles. We keep an explicit list of the live vertices. Bucket records are created when they are first touched, in a balanced search tree keyed by $(u,j)$. Hence, no table over all pairs $(u,j)$ needs to be cleared, and all dictionaries are deterministic. The entries of each ball are stored in an append-only array, so entry~$j$ of $L(u)$ is accessed in constant time, and a balanced search tree supports membership queries and entry lookups by center. Finally, in the comparison-addition model, the cutoff test $s<U/2$ is implemented as $s+s<U$.

\begin{proof}[Proof of \Cref{lem:one-level}]
By \Cref{lem:batched-growth}, growth calls satisfy (G1)--(G3). Therefore, \Cref{lem:selection} gives $|R|\leq r$, $O(\log n)$ rounds, $X=O(nN/r)$, and the bounds on~$S_t$. By \Cref{lem:batched-growth}, we get that $\sum_v|B_{A,R}^{<U_{\rm out}/2}(v)|\leq X=O(nN/r)$. \Cref{lem:batched-growth} also shows that $U_{\rm out}\leq U_{\rm in}$ and that every finite value comes with a cycle.

Next, we bound the work over all growth calls of the level. The active vertices of round~$t$ belong to~$S_t$, so inserting and removing incoming edges takes $O(\sum_t\sum_{v\in S_t}\deg(v))=O(m)$ heap operations. Similarly, maintaining the list of live vertices takes $O(\sum_t|S_t|)=O(n)$ time.

Every processed event either inserts an entry, or is a tree-edge duplicate, or closes a cycle. At most $X$ events insert entries. The center trees have fewer than~$X$ edges in total, each edge has two orientations, and each orientation is processed at most once per level for its center. Therefore, at most $2X$ events are tree-edge duplicates. At most one event per call closes a cycle and ends the call, and the subsequent cycle extraction costs $O(n)$ time. Each processed event and each new entry changes at most two bucket representatives, and all heaps and dictionaries take $O(\log n)$ time per operation. Hence, all growth calls together take $O((m+X)\log n+n\log n)$ time, and selection takes $O((nN/r)\log n)$ time. Since $X=O(nN/r)$ and $N/r\geq1$, we get the claimed running time of $O((m+nN/r)\log n)$, as required.
\end{proof}

\fi

\subsection{Constructing the hierarchy and completing the algorithm}
\label{sec:deterministic-completion}

\Cref{alg:main} constructs the hierarchy level by level with \Cref{lem:one-level}, using a single global best certificate, and then runs $\BoundedCycle$, which performs new cluster searches on the completed hierarchy; the balls and trees used to choose the levels are discarded. For $\kappa=1$ no level is constructed, $U_0=+\infty$, and $\BoundedCycle$ is the unbounded algorithm on the hierarchy $A_0=V$, $A_1=\emptyset$.

\begin{algorithm2e}[htbp]
\caption{Deterministic girth approximation}\label{alg:main}
\KwIn{Graph $G=(V,E,\ell)$ and an integer $\kappa\geq1$}
$U\leftarrow+\infty$ with no certificate; $q\leftarrow n^{1/\kappa}$; $A_0\leftarrow V$; $A_j\leftarrow\emptyset$ for $1\leq j\leq \kappa$\;
\For{$i=0,\ldots,\kappa-2$}{
 $N\leftarrow|A_i|$\;
 \lIf{$N<q$}{\textbf{break}}
 Apply \Cref{alg:one-level} to $A=A_i$ with $r=\lfloor N/q\rfloor$ and the current certificate; set $A_{i+1}\leftarrow R$ and update $U$\;
 \lIf{$A_{i+1}=\emptyset$}{\textbf{break}}
}
$U_0\leftarrow U$; $\rho_0\leftarrow U_0/2$\;
Run $\BoundedCycle$ on $A_0\supseteq\cdots\supseteq A_\kappa$ with the certificate of value $U_0$\;
\KwRet{the cycle of the best certificate}
\end{algorithm2e}

\begin{lemma}[Hierarchy]\label{lem:deterministic-hierarchy}
The construction of the hierarchy in \Cref{alg:main} takes $O((m+nq)\kappa\log n)$ time. The hierarchy satisfies
\begin{equation}\label{eq:hierarchy-bounds}
 |A_i|\leq n/q^i\quad(0\leq i<\kappa),
 \qquad
 \sum_{i=0}^{\kappa-1}\sum_{v\in V}|B_{A_i,A_{i+1}}^{<\rho_0}(v)|=O(\kappa nq).
\end{equation}
\end{lemma}
\begin{proof}
At each constructed level, $|A_{i+1}|\leq\lfloor|A_i|/q\rfloor$. If $N\geq q$, then $r\geq1$ and $N/r\leq2q$, so \Cref{lem:one-level} bounds the level-$i$ truncated bunches by $O(nq)$ at that level's output value $U_{\rm out}$, in time $O((m+nq)\log n)$. Since $U$ only decreases afterwards, $\rho_0\leq U_{\rm out}/2$, so these bunches only shrink. If $N<q$, or if $i=\kappa-1$ (then $|A_{\kappa-1}|\leq q$, also for $\kappa=1$ where $q=n$), the level-$i$ truncated bunches have total size at most $n|A_i|\leq nq$. Levels after an empty level contribute nothing. Summing over the at most $\kappa$ levels proves the claim.
\end{proof}

\begin{proof}[Proof of \Cref{thm:deterministic-girth}]
\Cref{lem:truncation-transfer} gives the approximation guarantee. The hierarchy costs $O((m+nq)\kappa\log n)$ time by \Cref{lem:deterministic-hierarchy}.

In $\BoundedCycle$, the preprocessing computes $p_i(v)$ and $\delta(v,A_i)$ for every level by a single run of Dijkstra's algorithm from a dummy vertex joined to all of~$A_i$ by edges of length~$0$, and builds the edge-access structure; this takes $O((m+n)\kappa\log n)$ time. By \Cref{lem:bounded-search}, the search from $a\in A_i\setminus A_{i+1}$ takes $O((1+|\Cl_{A_i,A_{i+1}}^{<\rho_0}(a)|)\log n)$ time. Hence, by double counting and \Cref{lem:deterministic-hierarchy}, all bounded cluster searches together take $O(\kappa nq\log n)$ time. The edge-and-pivot checks cost $O(\kappa m\log n)$ time with deterministic dictionaries. For the certificates found by the final searches, we keep only a description of the best one and extract its cycle at the end in $O(n)$ time. Substituting $q=n^{1/\kappa}$ proves the running-time bound.
\end{proof}

\subsection{Stretch \texorpdfstring{$2k/3$}{2k/3} for every \texorpdfstring{$k$}{k}}
\label{sec:odd}

For even~$k$, \Cref{thm:deterministic-girth} with parameter $\kappa=k/2$ returns a cycle of length at most $2kg/3$ in $O((m+n^{1+2/k})k\log n)$ time. In the rest of this section, $k=2\alpha+1\geq3$ is odd and $q=n^{2/k}$. 

\paragraph{The randomized algorithm.}
The algorithm of~\cite[Algorithm~1]{kadria2026tighter} uses a hierarchy $V=A_0\supseteq A_1\supseteq\cdots\supseteq A_\alpha\supseteq A_{\alpha+1}=\emptyset$, where $A_1$ is sampled from~$V$ with probability $n^{-1/k}$ and $A_{i+1}$ is sampled from~$A_i$ with probability $n^{-2/k}$ for $1\leq i<\alpha$. Phase~1 runs, for every $u\in V$ with $B_0(u)=\{v:\delta(u,v)<\delta(u,A_1)\}\neq\emptyset$, a Dijkstra-based procedure from~$u$ in $G[B_0(u)]$, in $O(|B_0(u)|^2\log n)$ time; if $G[B_0(u)]$ contains a shortest cycle of~$G$ through~$u$, it returns a cycle of length~$g$~\cite[Section~2]{kadria2026tighter}. Phase~2 runs $\ClusterOrCycle(a)$ for every $a\in A_1$, and Phase~3 performs the edge-and-pivot checks. Randomness is used for hashing and to bound $\sum_u|B_0(u)|^2$ and the cluster sizes of levels $1,\ldots,\alpha$ in expectation.

\begin{lemma}[First level]\label{lem:first-level}
A set $A_1\subseteq V$ with $|A_1|\leq n^{1-1/k}$, the sets $B_0(u)$ for all $u\in V$, and Phase~1 on them can be computed deterministically in $O((m+n^{1+2/k})\log n)$ time.
\end{lemma}
\begin{proof}
First sort the adjacency lists by length, breaking ties by neighbor identifier, and build a balanced search tree of the edges indexed by their endpoints. This one-time preprocessing costs $O(m\log n)$ and accounts for the $m$ term.

Apply Yan's algorithm~\cite[Algorithm~1 and Theorem~4]{yan2026lossless} with $p=2$ and $r=\lfloor n^{1-1/k}\rfloor$, growing exact balls of nearest vertices.  On sorted adjacency lists, a ball of $b$ nearest vertices and their distances can be computed in $O(b^2)$ time~\cite[Section~2]{DBLP:conf/soda/ChechikLRSTW14}. Since quotas double while the number of unhit balls decreases by a factor of eight, the total time for these computations, including repeated growth requests, is $O(n(n/r)^2)=O(n^{1+2/k})$. Yan's theorem gives $|A_1|\leq r$ and final balls $L(u)$ with $\sum_u|L(u)|^2=O(n^{1+2/k})$, each containing a vertex of~$A_1$.

Thus $B_0(u)\subseteq L(u)$. We extract $B_0(u)$ in $O(|L(u)|)$ time as the set of vertices $v\in L(u)$ with $\delta(u,v)<\min\{\delta(u,a):a\in A_1\cap L(u)\}$. We obtain $G[B_0(u)]$ by looking up every pair of its vertices in the edge search tree, then run Phase~1 on each nonempty $B_0(u)$. Constructing these subgraphs and running Phase~1 costs $O(\sum_u|B_0(u)|^2\log n)=O(n^{1+2/k}\log n)$ in total. Adding preprocessing proves the claim.
\end{proof}

\paragraph{The other levels.}
For $i=1,\ldots,\alpha-1$, we construct $A_{i+1}$ from~$A_i$ with \Cref{lem:one-level} and $r=\lfloor|A_i|/q\rfloor$, as in \Cref{alg:main}, sharing the best certificate and stopping when $|A_i|<q$ or $A_{i+1}=\emptyset$. Then $|A_\alpha|\leq n^{1-1/k}/q^{\alpha-1}=q$. Let $U_0$ be the final certificate value and $\rho_0=U_0/2$. We then run Phases~2 and~3 with every cluster search replaced by $\ClusterOrCycleBounded(a,\rho_0)$ and all hash tables replaced by balanced search trees. 

Fix a shortest cycle~$C$. Let $h_i(C)=\min_{w\in A_i}\delta(w,C)$, and, if $M(C)>g/3$, let $h'_i(C)=\min\{\delta(u,A_i),\delta(u',A_i)\}$ for a longest edge $\{u,u'\}$ of~$C$. We write $\eta_i$ for $h_i(C)$ if $M(C)\leq g/3$ and for $h'_i(C)$ otherwise.

\begin{restatable}[Odd~$k$]{lemma}{oddtransferlemma}\label{lem:odd-transfer}
Let $k=2\alpha+1\geq3$, fix a hierarchy $V=A_0\supseteq\cdots\supseteq A_{\alpha+1}=\emptyset$, and suppose that Phase~1 is run on the exact sets $B_0(u)$ and that every cluster search for a center of level $i\geq1$ is $\Delta$-good for every $\Delta\leq kg/3$. Then the certificates found by the three phases include one of value at most $2kg/3$.
\end{restatable}
\ifC
\else
\begin{proof}
Let $C$ be the shortest cycle fixed in \Cref{sec:odd}, and let $\eta_i$ be as defined there. The analysis of~\cite{kadria2026tighter} has two steps.
\begin{enumerate}
 \item[(1)] Either some certificate has value at most $2g$, or $\eta_2\leq g$~\cite[Lemma~10 and the proofs of Lemmas~11 and~12]{kadria2026tighter}. This step uses Phase~1 and one cluster search at level~$1$, from $a=p_1(u)$ for a vertex $u\in C$, with radius $\Delta\leq g$.
 \item[(2)] For $2\leq i\leq\alpha$, either some certificate has value at most $2\eta_i+4g/3$, or $\eta_{i+1}\leq\eta_i+2g/3$~\cite[Claims~11.1 and~12.1]{kadria2026tighter}. These proofs are the inductive step in the proof of \Cref{lem:krsvz-analysis} with $\Delta_{i-1}$ replaced by~$\eta_i$. They use only that one search at level~$i$ is $\Delta$-good for $\Delta=\eta_i+2g/3$ if $M(C)\leq g/3$, and for $\Delta=\eta_i+g-M(C)<\eta_i+2g/3$ if $M(C)>g/3$. 
\end{enumerate}
Suppose that no certificate has value at most $2kg/3$. By~(1), $\eta_2\leq g$. By induction, $\eta_i\leq g+2(i-2)g/3$ for $2\leq i\leq\alpha+1$: if (2) gave a certificate at level $i\leq\alpha$, its value would be at most $2g+4(i-2)g/3+4g/3\leq2g+4(\alpha-1)g/3=2kg/3$. Hence $\eta_{\alpha+1}<\infty$, contradicting $A_{\alpha+1}=\emptyset$. All radii used are at most $\max\{g,\ g+2(\alpha-1)g/3\}=kg/3$, so the searches involved are $\Delta$-good by assumption.

It remains to prove~(1) under our assumptions, since~\cite[Lemma~10]{kadria2026tighter} relies on~\cite[Lemma~5]{kadria2026improvedsicomp}. 

Let $u\in C$, and let $\{x,y\}$ be the farthest edge of~$C$ from~$u$ with $\delta(u,x)\geq\delta(u,y)$, let $P(u,x)$ be the shortest $u$--$x$ path along~$C$, and let $P_2(u,x)$ be the other $u$--$x$ path along~$C$. If $C\subseteq G[B_0(u)]$, then Phase~1 returns a cycle of length~$g\leq2g$. Otherwise, consider the search from $a=p_1(u)$ with radius $\Delta=\delta(u,A_1)+\ell(P_2(u,x))\leq g$. If $a\in A_2$, then directly $\delta(y,A_2)\leq\delta(u,A_1)+\delta(u,y)\leq g$, since $\delta(u,A_1)\leq\delta(u,x)\leq g/2$ in this case, so we may assume that $a\in A_1\setminus A_2$.

Orient every edge of~$C$ away from~$u$, as in~\cite[Lemma~10]{kadria2026tighter}: orient each edge other than $\{x,y\}$ along a shortest path from~$u$, and orient $\{x,y\}$ as $(y,x)$. By the triangle inequality, every such orientation $(s,t)$ has key
\[
 \delta(a,s)+\ell(s,t)\leq\delta(u,A_1)+\delta(u,s)+\ell(s,t)\leq\delta(u,A_1)+\ell(P_2(u,x))=\Delta,
\]
since $\delta(u,s)+\ell(s,t)$ is at most $g/2\leq\ell(P_2(u,x))$ for the edges other than $\{x,y\}$, and at most the length of one of the two $u$--$x$ paths for $\{x,y\}$. If all these orientations are cluster orientations of $\Cl(a)$, then (C2) would put all edges of~$C$ into the tree~$T_a$, which is impossible, so (C1) gives a certificate of value at most $2\Delta\leq2g$. Otherwise, one of these orientations is not a cluster orientation, and this is the only property of the search that the remaining argument of~\cite[Lemma~10]{kadria2026tighter} uses. It shows that $\min\{\delta(x,A_2),\delta(y,A_2)\}\leq g$, and hence that $\eta_2\leq g$~\cite[proofs of Lemmas~11 and~12]{kadria2026tighter}.

\end{proof}

\fi

\begin{proof}[Proof of \Cref{thm:deterministic-2k3} for odd~$k$]
By \Cref{lem:first-level}, the first level and Phase~1 take $O((m+n^{1+2/k})\log n)$ time. As in the proof of \Cref{lem:deterministic-hierarchy}, the truncated bunches of levels $1,\ldots,\alpha$ have total size $\sum_{i=1}^{\alpha}\sum_v|B_{A_i,A_{i+1}}^{<\rho_0}(v)|=O(\alpha nq)$, and their construction takes $O((m+nq)\alpha\log n)$ time. As in the proof of \Cref{thm:deterministic-girth}, Phases~2 and~3 take $O((m+nq)\alpha\log n)$ time, so the total time is $O((m+n^{1+2/k})k\log n)$.

If $U_0\leq2kg/3$, the certificate of value~$U_0$ suffices. Otherwise $\rho_0>kg/3$, and by \Cref{cor:bounded-good}, every search $\ClusterOrCycleBounded(a,\rho_0)$ is $\Delta$-good for every $\Delta\leq kg/3<\rho_0$. Hence \Cref{lem:odd-transfer} gives a certificate of value at most $2kg/3$.
\end{proof}

\section*{Use of AI tools}
We used AI assistants based on large language models in this research, in each case the latest models available in October 2026 under a standard individual paid plan. ChatGPT (OpenAI), Claude (Anthropic) and Gemini (Google) were used to edit, write and revise the paper, to fix typos and to improve clarity. All models were also used to check for possible issues with correctness, omissions from the literature review, and inaccurate citations of previous results.
ChatGPT (OpenAI) was used to discuss the ideas and to better understand previous papers. Claude (Anthropic) was used to create the illustrations of the paper. 
The author verified all proofs and ideas in the paper, proofread the entire paper, and takes full responsibility for its content.

\bibliography{articles}

\end{document}